\documentclass[pra,twocolumn,superscriptaddress,
preprintnumbers,amsmath,amssymb,floatfix]{revtex4}

\usepackage{graphicx}
\usepackage{float}
\usepackage{subfigure}
\usepackage{amsthm}
\usepackage{tensor}
\usepackage{color}
\usepackage[all]{xy}
\usepackage{tikz}
\usepackage{dsfont}
\usepackage{times,txfonts}
\usetikzlibrary{positioning}
\usepackage{braket}

\newtheorem{Thm}{Theorem}

\newtheorem{Cor}[Thm]{Corollary}
\theoremstyle{definition}

\begin{document}

\title{Investigating controlled teleportation capability of quantum states with respect to $k$-separability}

\author{Minjin Choi}
\email{mathcmj89@gmail.com }
\affiliation{
Division of National Supercomputing, Korea Institute of Science and Technology Information, Daejeon 34141, Korea
}

\author{Jeonghyeon Shin} 
\affiliation{
Center for Quantum Information, Korea Institute of Science and Technology (KIST), Seoul, 02792, Republic of Korea
}
\affiliation{
Department of Mathematics and Research Institute for Basic Sciences,
Kyung Hee University, Seoul 02447, Korea}

\author{Gunho Lee}
\affiliation{
Department of Mathematics and Research Institute for Basic Sciences,
Kyung Hee University, Seoul 02447, Korea}

\author{Eunok Bae}
\email{eobae@kias.re.kr}
\affiliation{
School of Computational Sciences, Korea Institute for Advanced Study (KIAS), Seoul 02455, Korea}

\date{\today}

\begin{abstract}
Quantum teleportation is an essential application of quantum entanglement. 
The examination of teleportation fidelity in two-party standard teleportation schemes reveals a critical threshold distinguishing separable and entangled states. 
For separable states, their teleportation fidelities cannot exceed the threshold, emphasizing the significance of entanglement. 
We extend this analysis to multi-party scenarios known as controlled teleportation. 
Our study provides thresholds that $N$-qudit $k$-separable states cannot exceed in a controlled teleportation scheme,
where $N \ge 3$ and $2 \le k \le N$.
This not only establishes a standard for utilizing a given quantum state as a resource in controlled teleportation 
but also enhances our understanding of the influence of the entanglement structure on controlled teleportation performance.
In addition, we show that genuine multipartite entanglement is not a prerequisite for achieving a high controlled teleportation capability.
\end{abstract}
\maketitle

\maketitle

\section{Introduction}
\label{sec: Introduction}
Quantum teleportation is a crucial application of quantum entanglement, finding utility in quantum communication, quantum networks, and quantum computation. 
In the standard teleportation scheme, Alice can perfectly transfer quantum information to Bob by using a two-qubit maximally entangled state and two bits of classical information~\cite{BBC93}. 
Nonmaximally entangled states can also be considered in a standard teleportation scheme, 
and we can evaluate the quantum state's usefulness for teleportation through teleportation fidelity~\cite{HHH99, BHH00}. 
It is important to note that there exists a threshold that separable states cannot surpass. 
Therefore, if an entangled state exceeds this threshold, we can assert its usefulness in standard teleportation.

In multi-party scenarios, controlled teleportation is a notable consideration~\cite{KB98, HBB99, LJK05, LJK07, GYL08}.
Here we consider the following controlled teleportation scheme.
Initially, $N$-players share an $N$-qudit state
and $N-2$ out of $N$ players conduct orthogonal measurements on their respective systems.
Subsequently, the remaining players execute the standard teleportation over the resulting state.
The receiver then reconstructs the sender's quantum state with the $N-2$ players' measurement outcomes.
For example, suppose that Alice, Bob, and Charlie share a Greenberger-Horne-Zeilinger (GHZ) state~\cite{GHZ89},
$\ket{\rm{GHZ}}_{ABC}=\left(\ket{000}_{ABC}+\ket{111}_{ABC}\right)/\sqrt{2}$.
If Bob measures his system in the $X$ basis $\left\{\ket{0_{x}}, \ket{1_{x}}\right\}$, where $\ket{i_{x}}=(\ket{0}+(-1)^{i}\ket{1})/\sqrt{2}$,
then the rest players share a two-qubit maximally entangled state, enabling them to execute two-qubit standard teleportation completely.
Without Bob's measurement, 
Alice and Charlie cannot carry out teleportation surpassing the classical channel
since the reduced state $\rho_{AC}$ of the GHZ state is a separable state.
Even if Bob performs the measurement and does not disclose the measurement outcome, Alice and Charlie encounter difficulties in performing teleportation perfectly due to the dependency of the resulting state of Alice and Charlie on the measurement result of Bob.

The controlled teleportation capability in this scheme can be directly evaluated as the average of the two-party teleportation fidelities of the resulting quantum states obtained after the measurements~\cite{LJK05, LJK07}.
By noting that any orthogonal measurements on fully separable states fail to entangle the resulting states,
we can readily obtain a threshold that fully separable states cannot exceed.
However, in multi-party scenarios, considering separability involves not only fully separable states but also other types of separable states.
For instance, in the case of $\ket{\psi}_{ABC}=\left(\ket{000}_{ABC}+\ket{110}_{ABC}\right)/\sqrt{2}$, 
system $A$ is entangled with systems $BC$, while systems $AB$ and $C$ are separable. 
In more detail, for $N \ge 3$ and $2 \le k \le N$,
stating that an $N$-partite pure state is $k$-separable 
means that the quantum state can be separated into $k$ parts. 
An $N$-partite mixed state is $k$-separable 
if it can be represented as a convex sum of $k$-separable pure states.
If an $N$-partite quantum state does not correspond to any case of separability, 
then it is termed as a genuinely multipartite entangled state.  
Therefore, refining the threshold for the controlled teleportation capability based on $k$-separability is significant for utilizing a quantum state as a resource in controlled teleportation
and for understanding the effect of the entanglement structure of quantum states in controlled teleportation.

In this work, we provide thresholds that $k$-separable states cannot surpass in the controlled teleportation scheme, 
thereby elucidating the impact of the entanglement structure on the controlled teleportation performance. 
We also demonstrate that values equal to the thresholds can be achieved using $k$-separable states, 
indicating that our thresholds are tight criteria.
By using these results, we observe that a high controlled teleportation capability can be attained by a biseparable state,
suggesting that genuine multipartite entanglement is not a precondition for achieving a high controlled teleportation capability.
However, we also show that a high controlled teleportation capability is achievable only when the number of entangled systems in the pure states constituting the biseparable state is sufficiently close to $N$.

This paper is organized as follows. 
We introduce the maximal average teleportation fidelity of the controlled teleportation over an $N$-qudit state in Sec.~\ref{sec: telecapa}
and investigate the controlled teleportation capability in Sec.~\ref{sec: usefulness}.
An example with the isotropic GHZ state is provided in Sec.~\ref{sec: isotropicGHZ}, 
followed by a summary and discussion of our results in Sec.~\ref{sec: conclusion}.

\section{Teleportation capability in controlled teleportation} 
\label{sec: telecapa}

We begin by reviewing the maximal fidelity of teleportation over a two-qudit state $\rho_{AB} \in \mathbb{C}^{d}_{A} \otimes \mathbb{C}^{d}_{B}$.
The teleportation fidelity is defined as~\cite{P94}
\begin{equation}
F\left(\Lambda_{\rho_{AB}}\right)=\int d\xi\bra{\xi}\Lambda_{\rho_{AB}}\left(\ket{\xi}\bra{\xi}\right)\ket{\xi},
\end{equation}
where $\Lambda_{\rho_{AB}}$ represents a given teleportation scheme  over a two-qudit state $\rho_{AB}$
and the integral is taken over a uniform distribution with respect to all one-qudit pure states.
It is worth noting that the teleportation fidelity is associated with the fully entangled fraction of $\rho_{AB}$~\cite{BDS96}, 
defined as
\begin{equation}
\label{fraction}
f\left(\rho_{AB}\right)=\rm{max}\bra{e}\rho_{AB}\ket{e},
\end{equation}
where the maximum is taken over all maximally entangled states $\ket{e}$ of two-qudit systems.
When $\bar{\Lambda}_{\rho_{AB}}$ represents the standard teleportation scheme over $\rho_{AB}$ to achieve maximal fidelity, 
it has been proven that the following equation holds~\cite{HHH99, BHH00}: 
\begin{equation}
\label{two-qudit_fidel}
F\left(\bar{\Lambda}_{\rho_{AB}}\right)=\frac{df\left(\rho_{AB}\right)+1}{d+1}.
\end{equation}
The given state $\rho_{AB}$ is said to be useful for teleportation 
if $F\left(\bar{\Lambda}_{\rho_{AB}}\right)>2/(d+1)$,
as the classical channel can have at most $F=2/(d+1)$~\cite{HHH96, HHH99}.
In other words, two-qudit separable states cannot exhibit a teleportation fidelity $F$ exceeding $2/(d+1)$.

Let us now consider the following controlled teleportation.
Suppose that $N \ge 3$ players, $A_{1}$, $A_{2}$, ..., $A_{N}$, share an $N$-qudit state $\rho_{A_{1} \cdots A_{N}} \in \mathbb{C}^{d}_{A_{1}} \otimes \cdots \otimes \mathbb{C}^{d}_{A_{N}}$.
Following individual orthogonal measurements on their respective systems by all players except $A_{i}$ and $A_{j}$,
the players $A_{i}$ and $A_{j}$ carry out the standard teleportation over the resulting state with the measurement outcomes.
We note that a one-qudit orthogonal measurement can be described in terms of the computational basis $\{\ket{k}\}_{k=0}^{d-1}$ and $d \times d$ unitary operator $U$.
Thus we can say that
the players $A_{i}$ and $A_{j}$ 
have the resulting state $\sigma_{A_{i}A_{j}}^{U_{K_{ij}}, J}$
with probability $\bra{J}U_{K_{ij}}\rho_{K_{ij}}U_{K_{ij}}^{\dagger}\ket{J}$
after the $N-2$ players' measurements.
Here,
$K_{ij}=\{A_{k_{1}}, A_{k_{2}}, \dots, A_{k_{N-2}}\}=\{A_{1},A_{2}, \dots, A_{N}\} \setminus \{A_{i}, A_{j}\}$,
$U_{K_{ij}}=U_{A_{k_{1}}} \otimes U_{A_{k_{2}}} \otimes \cdots \otimes U_{A_{k_{N-2}}}$ is a tensor product of local unitary operators,
$J \in \mathbb{Z}_{d}^{N-2}$ is the $N-2$ players' measurement outcome,
and
$\rho_{K_{ij}}=\rho_{A_{k_{1}}} \otimes \rho_{A_{k_{2}}} \otimes \cdots \otimes \rho_{A_{k_{N-2}}}$ with the reduced state $\rho_{A_{k_{t}}}$ on system $A_{k_{t}}$.
Then, the maximal average teleportation fidelity $F_{ij}^{(N)}$ is naturally derived as follows~\cite{LJK05, LJK07}:
\begin{equation}
F_{ij}^{(N)}\left(\rho_{A_{1} \cdots A_{N}}\right)
=\max_{U_{K_{ij}}}\sum_{J\in \mathbb{Z}_{d}^{N-2}}
\bra{J}U_{K_{ij}}\rho_{K_{ij}}U_{K_{ij}}^{\dagger}\ket{J}
F\left(\bar{\Lambda}_{\sigma_{A_{i}A_{j}}^{U_{K_{ij}}, J}}\right).
\end{equation}
By linearity and Eq.~(\ref{two-qudit_fidel}), we obtain
\begin{equation}
\label{eq:tele_fidel}
F_{ij}^{(N)}\left(\rho_{A_{1} \cdots A_{N}}\right)=\frac{df_{ij}^{(N)}\left(\rho_{A_{1} \cdots A_{N}}\right)+1}{d+1},
\end{equation}
where
\begin{equation}
\label{eq:tele_frac_N}
f_{ij}^{(N)}(\rho_{A_{1} \cdots A_{N}})
=\max_{U_{K_{ij}}}\sum_{J\in \mathbb{Z}_{d}^{N-2}}
\bra{J}U_{K_{ij}}\rho_{K_{ij}}U_{K_{ij}}^{\dagger}\ket{J}
f\left(\sigma_{A_{i}A_{j}}^{U_{K_{ij}}, J}\right).
\end{equation}

To assess the overall controlled teleportation capability of $\rho_{A_{1} \cdots A_{N}}$,
we consider the minimum of $F_{ij}^{(N)}$, where the minimum is taken over all distinct $i$ and $j$.
The minimum of $F_{ij}$ not only implies that $\rho_{A_{1} \cdots A_{N}}$ has at least that much controlled teleportation capability for any arbitrary $i$ and $j$ according to its own definition,
but it can also be anticipated to be an important quantity for investigating the relationship with genuine multipartite entanglement.
Indeed, it is related to GHZ distillability~\cite{LJK07} and can be used to define a genuine multipartite entanglement measure~\cite{CBL23}.

We can assert that $\rho_{A_{1} \cdots A_{N}}$ is more useful for controlled teleportation than fully separable states if $\min\{F_{ij}^{(N)}\}>2/(d+1)$, where the minimum is taken over all different $i$ and $j$.
This assertion stems from the observation that if $\rho_{A_{1} \cdots A_{N}}$ is a fully separable state, then the resulting state after $N-2$ players' measurements is a two-qudit separable state.
However, when considering multi-party scenarios, different types of separable states can also be taken into account.
In the following section, we investigate the controlled teleportation capability of quantum states in the controlled teleportation scheme based on the degree of separability.

\section{Investigating the controlled Teleportation capability} 
\label{sec: usefulness}
We first introduce some definitions and notations.
For an $N$-partite pure state $\ket{\psi}$, where $N \ge 3$ and $2 \le k \le N$, it is called $k$-separable 
if there exists a partition $\mathcal{P}_{k}=\{ P_{1}, P_{2}, ..., P_{k}\}$ of the set 
$\{A_{1}, A_{2}, ..., A_{N}\}$ such that 
it can be written as $\ket{\psi}=\ket{\psi_{P_{1}}} \otimes \ket{\psi_{P_{2}}} \otimes \cdots \otimes \ket{\psi_{P_{k}}}$.
In this case, we denote $ \ket{\psi} \in SEP(k)$.
When emphasizing the partition $\mathcal{P}_k$, we write
$\ket{\psi} \in SEP\left(\mathcal{P}_{k}\right)$.
An $N$-partite mixed state $\rho$ is called $k$-separable if it can be written as a convex sum of $k$-separable pure states $\rho=\sum_{l}p_{l}\ket{\psi_{l}}\bra{\psi_{l}}$, where the $k$-separable pure states $\{\ket{\psi_{l}}\}$ can be $k$-separable with respect to different partitions,
and we denote $\rho \in SEP(k)$.
In particular, if there exists a partition $\mathcal{P}_{k}$ such that $\ket{\psi_{l}} \in SEP\left(\mathcal{P}_{k}\right)$ for all $l$, then let us denote $\rho \in SEP\left(\mathcal{P}_{k}\right)$.
We note that $SEP(k) \subset SEP(k-1)$ for $3 \le k \le N$.  For an $N$-partite state $\rho$, it is called fully separable if $\rho \in SEP(N)$, and biseparable if $\rho \in SEP(2)$.
If an $N$-partite state is not biseparable,
then it is a genuinely $N$-partite entangled state.

The controlled teleportation capability of quantum states with respect to $k$-separability is under investigation.
Let $\rho_{A_{1} \cdots A_{N}}$ be an $N$-qudit $k$-separable state. 
By definition, it can be written as
\begin{equation}
\label{rep_sep}
\rho_{A_{1} \cdots A_{N}} = \sum_{\mathcal{P}_{k}}\alpha_{\mathcal{P}_{k}}\rho_{\mathcal{P}_{k}},
\end{equation}
where the sum is over all partitions $\mathcal{P}_{k}$ of size $k$, $\alpha_{\mathcal{P}_{k}} \ge 0$, and $\rho_{\mathcal{P}_{k}} \in SEP(\mathcal{P}_{k})$.
For a partition $\mathcal{P}_{k}$, 
if $\{A_{i}, A_{j}\} \in \mathcal{P}_{k}$,
then $f_{ij}^{(N)}\left(\rho_{\mathcal{P}_{k}}\right)$ is at most 1,
and otherwise, $f_{ij}^{(N)}\left(\rho_{\mathcal{P}_{k}}\right) \le 1/d$.
Hence, we obtain that
\begin{eqnarray}
\label{f_ij_upper}
f_{ij}^{(N)}(\rho_{A_1\cdots A_N}) 
&\le&
\sum_{\mathcal{P}_{k}}\alpha_{\mathcal{P}_{k}}f_{ij}^{(N)}(\rho_{\mathcal{P}_{k}}) \nonumber\\
&\le& \frac{1}{d}+\frac{d-1}{d}\sum_{\mathcal{P}_{k}: \{A_{i},A_{j}\} \in \mathcal{P}_{k}}\alpha_{\mathcal{P}_{k}}.
\end{eqnarray}
We now derive the following theorem.
\begin{Thm}
\label{thm1}
Let
\begin{equation}
\label{threshold}
T(d, N, k) = \frac{2}{d+1}+\frac{d-1}{d+1}\frac{(N-k+1)(N-k)}{N(N-1)}.
\end{equation}
For an $N$-qudit state $\rho_{A_1\cdots A_N}$,
if $\rho_{A_1\cdots A_N} \in SEP(k)$, then
\begin{equation}
\label{eq_thm2}
\min_{1\le i<j\le N}F_{ij}^{(N)}\left(\rho_{A_{1}\cdots A_{N}}\right) \le T(d, N, k).
\end{equation}
\end{Thm}
\begin{proof}
Assume that $\rho_{A_1\cdots A_N}$ is a $k$-separable state
of the form in Eq~(\ref{rep_sep}).
Then, by Eq.~(\ref{f_ij_upper}), $f_{ij}^{(N)}\left(\rho_{A_{1}\cdots A_{N}}\right)$
is upper bounded by
\begin{equation}
\frac{1}{d} + \frac{d-1}{d}\sum_{\mathcal{P}_{k}: \{A_{i},A_{j}\} \in \mathcal{P}_{k}}\alpha_{\mathcal{P}_{k}}.
\end{equation}
Since we consider all $i$ and $j$ such that $1 \le i < j \le N$,
there are $\binom{N}{2}$ upper bounds.
For a partition $\mathcal{P}_{k}$,  
$\alpha_{\mathcal{P}_{k}}$ appears in the upper bound of $f_{pq}^{(N)}\left(\rho_{A_{1}\cdots A_{N}}\right)$
only when $\{A_{p},A_{q}\} \in \mathcal{P}_{k}$.
Hence, the number of upper bounds where $\alpha_{\mathcal{P}_{k}}$ appears is $\sum_{t=1}^{k}\binom{|P_{t}|}{2}$, 
where $P_{t} \in \mathcal{P}_{k}$ and $|S|$ is the size of the set $S$.
Therefore,
\begin{eqnarray}
\label{thm2_proof_eq}
\sum_{1\le i<j\le N}f_{ij}^{(N)}\left(\rho_{A_{1}\cdots A_{N}}\right) 
&\le& \frac{1}{d}\binom{N}{2}+\frac{d-1}{d}\sum_{\mathcal{P}_{k}}\sum_{t=1}^{k}\binom{|P_{t}|}{2}\alpha_{\mathcal{P}_{k}} \nonumber \\
&\le& \frac{1}{d}\binom{N}{2}+\frac{d-1}{d}\binom{N-k+1}{2}\sum_{\mathcal{P}_{k}}\alpha_{\mathcal{P}_{k}} \nonumber \\
&=& \frac{1}{d}\binom{N}{2}+\frac{d-1}{d}\binom{N-k+1}{2}, 
\end{eqnarray}
where the second inequality is derived from the fact that $\sum_{t=1}^{k}\binom{|P_{t}|}{2}$ is maximized when one of the sets in partition $\mathcal{P}_{k}$ has a size of $N-k+1$ and the sizes of the remaining sets are all one.
The proof is completed by dividing both sides by $\binom{N}{2}$
since the minimum cannot exceed the average.
\end{proof}

\begin{figure}[t]
\label{fig_1}
\centering
\includegraphics[width=8.5cm]{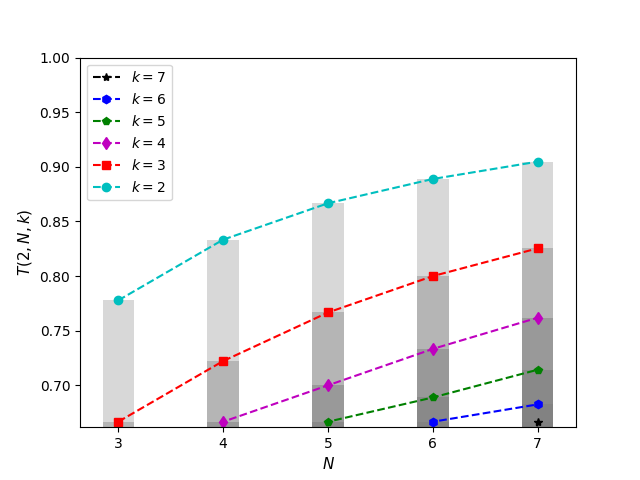}
\caption{
For $3 \le N \le 7$ and $2 \le k \le N$, $T(2, N, k)$ is expressed.
In any case, if the minimum of $F_{ij}^{(N)}$ for a given quantum state is greater than $2/3$, then the quantum state is more useful than fully separable states in controlled teleportation.
However, as $N$ increases, the threshold $T(2, N, 2)$ that a genuinely multipartite entangled state must surpass to outperfom any biseparable state in controlled teleportation also increases.
}
\end{figure}

Theorem~\ref{thm1} introduces the limitations of controlled teleportation capability achievable through $N$-qudit $k$-separable states, as depicted in FIG.~1.
Remark that when considering only pure states, 
$T(d, N, k)$ is a loose bound 
since 
\begin{equation}
\min_{1\le i<j\le N}F_{ij}^{(N)}\left(\ket{\psi}_{A_{1}\cdots A_{N}}\right)=\frac{2}{d+1}
\end{equation}
for any $N$-qudit biseparable pure state $\ket{\psi}_{A_{1}\cdots A_{N}}$.
However, there is an $N$-qudit $k$-separable mixed state 
that attains $T(d, N, k)$ as the minimum of $F_{ij}^{(N)}$ as follows below. 
In other words, $T(d, N, k)$ can be considered as a tight bound 
if we take account of all quantum states.

\begin{Thm}
\label{thm2}
There exists an $N$-qudit $k$-separable state $\rho$
such that 
\begin{equation}
\min_{1\le i<j\le N}F_{ij}^{(N)}\left(\rho\right) = T(d, N, k).
\end{equation}
\end{Thm}
\begin{proof}
For $t \in \mathbb{Z}_{d}$, let 
\begin{equation}
\ket{\phi_{M, t}} \equiv \frac{1}{\sqrt{d^{M-1}}}\sum_{i_{1}+\cdots+i_{M} \equiv t\pmod{d}}\ket{i_{1} \cdots i_{M}},
\end{equation}
where $i_{j} \in \mathbb{Z}_{d}$.
Consider the following symmetric $N$-qudit $k$-separable state
\begin{equation}
\label{thm3_state}
\rho = \frac{1}{\binom{N}{k-1}}\sum_{\mathcal{A}_{k-1}} 
\ket{0 \cdots 0}_{\mathbf{A}_{k-1}}\bra{0 \cdots 0}
\otimes \ket{\phi_{N-k+1, 0}}_{\bar{\mathbf{A}}_{k-1}}\bra{\phi_{N-k+1, 0}},
\end{equation}
where the sum is over all subsets $\mathcal{A}_{k-1}$ of $\mathcal{A}=\{A_{1}, A_{2}, ..., A_{N}\}$ with size $k-1$,
$\mathbf{A}_{k-1}$ and $\bar{\mathbf{A}}_{k-1}$
correspond to the system of $\mathcal{A}_{k-1}$ and $\mathcal{A} \setminus \mathcal{A}_{k-1}$, respectively.

Assume that players $A_{3}, ..., A_{N}$ measure their respective systems in the computational basis $\{\ket{0}, \ket{1}, ... ,\ket{d-1}\}$,
and they get the outcome $i_{3} \cdots i_{N} \in \mathbb{Z}_{d}^{N-2}$ with $i_{3} + \cdots + i_{N} \equiv s \pmod{d}$.
Let $\sigma_{A_{1}A_{2}}^{(i_{3} \cdots i_{N})}$ be the resulting state of players $A_{1}$ and $A_{2}$.
From tedious but straightforward calculations,
we have that
\begin{eqnarray}
&&prob(i_{3} \cdots i_{N})
\bra{\phi_{2, 0}}\sigma_{A_{1}A_{2}}^{(i_{3} \cdots i_{N})}\ket{\phi_{2, 0}} \nonumber \\
&& =\frac{1}{\binom{N}{k-1}}\sum_{\mathcal{A}_{k-1}: \{A_{1}, A_{2}\} \not\subset \mathcal{A}\setminus\mathcal{A}_{k-1}} \prod_{j: A_{j} \in \mathcal{A}_{k-1}}\langle i_{j} | 0 \rangle  \frac{1}{d^{N-k+1}} \nonumber \\
&& +\frac{1}{\binom{N}{k-1}}\sum_{\mathcal{A}_{k-1}: \{A_{1}, A_{2}\} \subset \mathcal{A} \setminus \mathcal{A}_{k-1}} \prod_{j: A_{j} \in \mathcal{A}_{k-1}}\langle i_{j} | 0 \rangle  \frac{1}{d^{N-k-1}}
\end{eqnarray}
if $s \equiv 0\pmod{d}$, and
\begin{eqnarray}
&&prob(i_{3} \cdots i_{N})
\bra{\phi_{2, d-s}}\sigma_{A_{1}A_{2}}^{(i_{3} \cdots i_{N})}\ket{\phi_{2, d-s}} \nonumber \\
&&=
\frac{1}{\binom{N}{k-1}}\sum_{\mathcal{A}_{k-1}: |\{A_{1}, A_{2}\} \cap \mathcal{A}_{k-1}|=1} \prod_{j: A_{j} \in \mathcal{A}_{k-1}} \langle i_{j} | 0 \rangle  \frac{1}{d^{N-k+1}} \nonumber \\ 
&& +\frac{1}{\binom{N}{k-1}}\sum_{\mathcal{A}_{k-1}: \{A_{1}, A_{2}\} \subset \mathcal{A} \setminus \mathcal{A}_{k-1}} \prod_{j: A_{j} \in \mathcal{A}_{k-1}} \langle i_{j} | 0 \rangle  \frac{1}{d^{N-k-1}}
\end{eqnarray}
if $s \not\equiv 0 \pmod{d}$.
Hence, by the definition of $f_{ij}^{(N)}$,
\begin{eqnarray}
f_{12}^{(N)}(\rho) &\ge&
\sum_{i_{3}, \cdots, i_{N}} prob(i_{3} \cdots i_{N}) f\left(\sigma_{A_{1}A_{2}}^{(i_{3} \cdots i_{N})}\right) \nonumber \\
&\ge&
\frac{1}{d}\sum_{\mathcal{A}_{k-1}: \{A_{1}, A_{2}\} \not\subset \mathcal{A}\setminus\mathcal{A}_{k-1}} \frac{1}{\binom{N}{k-1}} \nonumber \\
&&+\sum_{\mathcal{A}_{k-1}: \{A_{1}, A_{2}\} \subset \mathcal{A}\setminus\mathcal{A}_{k-1}} \frac{1}{\binom{N}{k-1}} \nonumber \\
&=&
\frac{1}{d}+\frac{d-1}{d}\sum_{\mathcal{A}_{k-1}: \{A_{1}, A_{2}\} \subset \mathcal{A}\setminus\mathcal{A}_{k-1}} \frac{1}{\binom{N}{k-1}}.
\end{eqnarray}
It follows from Eq.~(\ref{f_ij_upper}) that
\begin{eqnarray}
f_{12}^{(N)}(\rho) &=& \frac{1}{d}+\frac{d-1}{d}\sum_{\mathcal{A}_{k-1}: \{A_{1}, A_{2}\} \subset \mathcal{A}\setminus\mathcal{A}_{k-1}} \frac{1}{\binom{N}{k-1}} \nonumber \\
&=& \frac{1}{d}+\frac{d-1}{d}\frac{\binom{N-2}{k-1}}{\binom{N}{k-1}} \nonumber \\
&=& \frac{1}{d}\binom{N}{2}+\frac{d-1}{d}\binom{N-k+1}{2}
\end{eqnarray}
Therefore,
\begin{equation}
F_{12}^{(N)}\left(\rho_{A_{1}\cdots A_{N}}\right) = T(d, N, k),
\end{equation}
and the symmetry completes the proof.
\end{proof}

Let us think about genuinely $N$-partite entangled states.
Since
\begin{equation}
T(d,N,2)=1-\frac{2(d-1)}{N(d+1)}
\end{equation}
and $T(d,N,2)$ converges to 1 as $N$ goes to infinity,
Theorem~\ref{thm2} tells us that a high controlled teleportation capability can be attained by using a biseparable state for sufficiently large $N$.
In other words, genuine multipartite entanglement is not a prerequisite for achieving a high controlled teleportation capability.

However, when we examine the pure states constituting the biseparable state in Eq.~(\ref{thm3_state}), 
$N-1$ out of $N$ systems are entangled.
Considering this fact, 
one might wonder if achieving a high controlled teleportation capability 
requires a sufficient number of systems close to $N$ to be entangled in the pure states constituting the biseparable state.
To clarify this, we further refine the set of biseparable states based on the number of entangled subsystems within the states constituting the biseparable state,
and examine the corresponding threshold for the controlled teleportation capability.
The thresholds for these cases can be obtained by slightly modifying the proof of Theorem~\ref{thm1}.

\begin{figure}[t]
\label{fig_2}
\centering
\includegraphics[width=8.5cm]{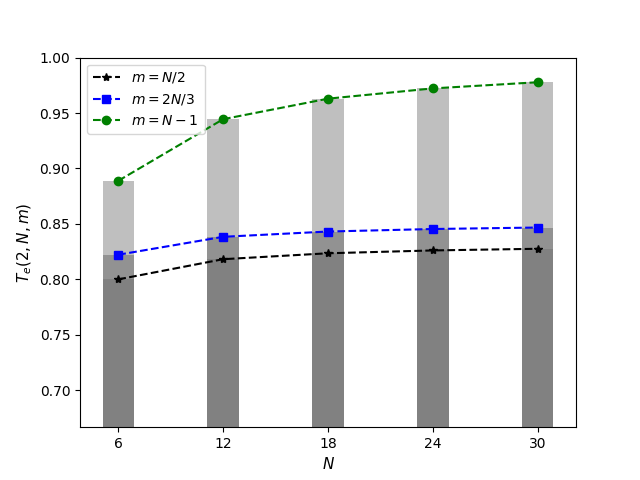}
\caption{
The upper bound $T_{e}(2, N, m)$ in Eq.~(\ref{cor_thsd}) is plotted, where $N=6t$ for $1 \le t \le 5$ and $m=N/2, 2N/3, N-1$.
As $N$ increases, $T_{e}(2, N, N-1)$ approaches $1$, while $T_{e}(2, N, N/2)$ and $T_{e}(2, N, 2N/3)$ cannot exceed 0.834 and 0.852, respectively.
}
\end{figure}

\begin{Cor}
\label{cor}
Suppose that the number of systems that can be entangled in pure states composing an ensemble of an N-qudit biseparable state $\rho$ is limited to $m$, where $m \ge \lceil N/2 \rceil$.
Then
\begin{equation}
\label{eq_cor}
\min_{1\le i<j\le N}F_{ij}^{(N)}\left(\rho_{A_{1}\cdots A_{N}}\right) \le T_{e}(d, N, m),
\end{equation}
where
\begin{equation}
\label{cor_thsd}
T_{e}(d, N, m) = 1 - \frac{2(d-1)m(N-m)}{(d+1)N(N-1)}
\end{equation}
\end{Cor}
\begin{proof}
In this case, we have
\begin{equation}
\sum_{t=1}^{2}\binom{|P_{t}|}{2} \le \binom{m}{2} + \binom{N-m}{2}
\end{equation}
in the first inequality in Eq.~(\ref{thm2_proof_eq}),
and thus,
\begin{equation}
\sum_{1\le i<j\le N}f_{ij}^{(N)}\left(\rho_{A_{1}\cdots A_{N}}\right) 
\le \frac{1}{d}\binom{N}{2}+\frac{d-1}{d}\left(\binom{m}{2} + \binom{N-m}{2}\right).
\end{equation}
We obtain the result through direct calculations.
\end{proof}

Let us consider scenarios where the proportion of systems that can be entangled in pure states composing an ensemble of an $N$-qudit biseparable state is limited by $\gamma$, where $1/2 \le \gamma \le 1$.
According to Corollary~\ref{cor}, the minimum of $F_{ij}^{(N)}$ is upper bounded by $T_{e}(d, N, \gamma N)$ as seen in FIG.~2.
Since $T_{e}(d, N, \gamma N)$ is an increasing function with respect to $N$,
\begin{equation}
\lim_{N \to \infty}T_{e}(d, N, \gamma N)=1 - \frac{2(d-1)\gamma(1-\gamma)}{d+1}
\end{equation}
serves as an upper bound of the minimum of $F_{ij}^{(N)}$.
This indicates that to achieve a high controlled teleportation capability, $\gamma$ should be close to $1$ at least.
Namely, the number of entangled subsystems in the pure states constituting the biseparable state impacts the controlled teleportation performance.


\section{Example: isotropic GHZ states} 
\label{sec: isotropicGHZ}

Suppose that $N \ge 3$ players $A_{1}, A_{2}, ..., A_{N}$ share an $N$-qubit isotropic GHZ state $\rho_{N}$ defined by
\begin{equation}
\label{isoGHZ}
\rho_{N} = p\ket{GHZ_{N}}\bra{GHZ_{N}}+\frac{1-p}{2^N}I_{N},
\end{equation}
where $\ket{GHZ_{N}} = \frac{1}{\sqrt 2}\left(\ket{0}^{\otimes N}+\ket{1}^{\otimes N}\right)$, $0<p<1$, and $I_{M}$ is the identity operator for $M$-qubit system.
If players $A_{3}, ..., A_{N}$ measure their systems in the $X$ basis $\{\ket{0_x},\ket{1_x}\}$, where $\ket{j_x}=\frac{1}{\sqrt2}\left(\ket{0}+(-1)^{j}\ket{1}\right)$ for $j=0, 1$, and they have the measurement outcomes $i_3,\dots ,i_N$,
then the players $A_{1}$ and $A_{2}$'s quantum state becomes
\begin{equation}
\sigma_{+}=p\ket{\phi^{+}}\bra{\phi^{+}} +\frac{1-p}{4}I_{2}
\end{equation}
if $i_3+\cdots +i_N \equiv 0\pmod 2$,
and
\begin{equation}
\sigma_{-}=p\ket{\phi^{-}}\bra{\phi^{-}} +\frac{1-p}{4}I_{2}
\end{equation}
if $i_3+\cdots +i_N \equiv 1\pmod 2$,
where $\ket{\phi^{\pm}}=\frac{1}{\sqrt 2}(\ket{00} \pm \ket{11})$.
Since
\begin{eqnarray}
f\left(\sigma_{\pm}\right)&\ge& \bra{\phi^{\pm}}\sigma_{\pm}\ket{\phi^{\pm}} \nonumber \\
&=& \frac{1}{4}+\frac{3}{4}p,
\end{eqnarray}
it follows from the definition of $f_{ij}$
\begin{equation}
f_{12}^{(N)}\left(\rho_{N}\right) \ge \frac{1}{4}+\frac{3}{4}p.
\end{equation}
On the other hand,
\begin{eqnarray}
f_{12}^{(N)}(\rho_{N}) 
&\le& p f_{12}^{(N)}\left(\ket{GHZ_{N}}\bra{GHZ_{N}}\right) \nonumber \\
&&+(1-p)f_{12}^{(N)}\left(\frac{1}{2^N}I_{N}\right) \nonumber \\
&=& \frac{1}{4}+\frac{3}{4}p.
\end{eqnarray}
Therefore, by symmetry,
\begin{equation}
\min_{1\le i<j\le N}F_{ij}^{(N)}\left(\rho_{N}\right) = \frac{1}{2}+\frac{1}{2}p.
\end{equation}

\begin{figure}[t]
\label{fig_3}
\centering
\includegraphics[width=8.5cm]{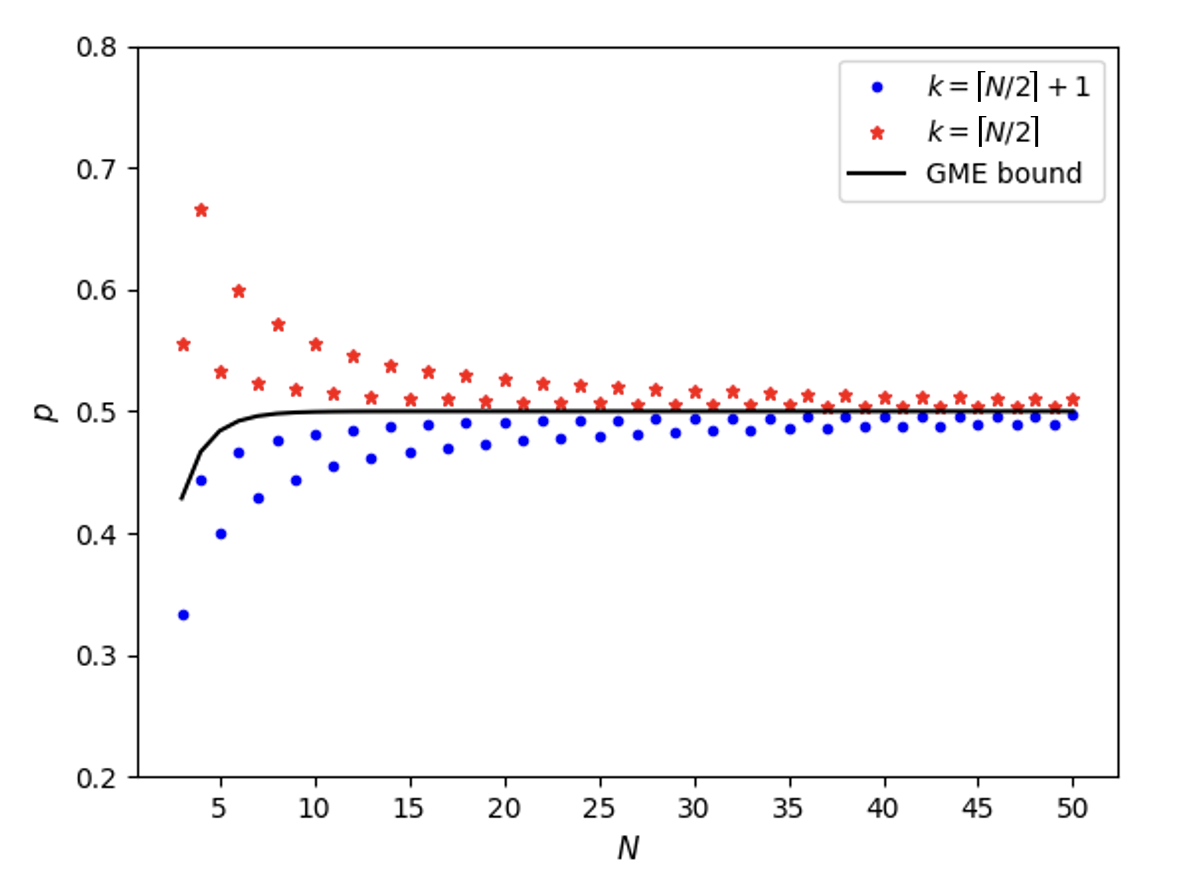}
\caption{
For the $N$-qubit isotropic GHZ state $\rho_{N}$ in Eq.~(\ref{isoGHZ}),
the value of $p$ for which the minimum of $F_{ij}^{(N)}$ becomes $T(2, N, k)$ is calculated, where $k=\lceil N/2 \rceil +1$ and $k= \lceil N/2 \rceil$.
The solid line represents the bound for being a genuinely $N$-partite entangled state, which is the right-hand side in the inequality in Eq.~(\ref{GMEbound}).
If $p$ is greater than the bound,
then $\rho_{N}$ is more useful than any $(\lceil N/2 \rceil +1)$-separable state in the controlled teleportation scheme.
However, it cannot be said that $\rho_{N}$ outperforms $\lceil N/2 \rceil$-separable states.
}
\end{figure}

Remark~\cite{YMM22} that
$\rho_{N}$ is a genuinely $N$-partite entangled state
if and only if 
\begin{equation}
\label{GMEbound}
p>\frac{2^{N-1}-1}{2^{N}-1}.    
\end{equation}
Hence, if $\rho_{N}$ is a genuinely $N$-partite entangled state,
then 
\begin{equation}
\min_{1\le i<j\le N}F_{ij}^{(N)}\left(\rho_{N}\right) > \frac{1}{2}+\frac{2^{N-1}-1}{2^{N+1}-2}.
\end{equation}
Theorem~\ref{thm1} implies that
if the minimum of $F_{ij}^{(N)}(\rho_{N})$ exceeds $T(2, N, k)$,
then $\rho_{N}$ outperforms any $k$-separable state in the controlled teleportation scheme.
Numerical analysis confirms that
if $\rho_{N}$ is a genuinely $N$-partite entangled state,
then it is more useful than $(\left\lceil N/2\right\rceil+1)$-separable states in the controlled teleportation scheme, as illustrated in FIG.~3.
However, if $k \ge \left\lceil N/2\right\rceil$,
then examples demonstrating genuinely multipartite entangled states with lower controlled teleportation capability than $k$-separable states can be easily found using Theorem~\ref{thm2}.

\section{Conclusion} 
\label{sec: conclusion}

In this paper, we have investigated the controlled teleportation capability with respect to $k$-separability
and established thresholds that $k$-separable states cannot exceed.
These thresholds are tight criteria because they can be attained with $k$-separable states.
Based on these results, we have shown that genuine multipartite entanglement is not a precondition for achieving a high controlled teleportation capability.
We have also demonstrated that controlled teleportation performance depends on the number of entangled subsystems in the pure states composing the biseparable state.
As an instance, we have examined the $N$-qubit isotropic GHZ states, 
and have illustrated that if an $N$-qubit isotropic GHZ state is a genuinely multipartite entangled state, 
then it has controlled teleportation capability that outperforms $k$-separable states, where $k = \lceil N/2 \rceil +1$. 
However, there exists a genuinely $N$-partite entangled state with lower controlled teleportation capability than $k$-separable states, when $k \ge \lceil N/2 \rceil$.

Based on our findings, we can consider the following future works.
The activation of genuine multipartite entanglement is a fascinating feature~\cite{YMM22, PV22}.
There exists a biseparable state $\rho$ such that $\rho^{\otimes n}$ is a genuinely multipartite entangled state for some $n$.
With this feature,
investigating whether there exists a quantum state $\rho$ that enhances the controlled teleportation capability for some $n$ could be an interesting task.
It could also be intriguing to investigate if similar results can be obtained in a more general controlled teleportation scheme.
For example, we have considered cases where all players involved in the control part performed orthogonal measurements,
but one could also explore scenarios where general measurements are conducted.

We can also look into other multi-party quantum applications such as conference key agreement~\cite{CL05} or quantum secret sharing~\cite{HBB99}.
Indeed, for conference key agreement, multipartite private states capable of having a perfect key are genuinely multipartite entangled states~\cite{DBW21}.
For quantum secret sharing,
there are quantum states that can have a perfect key for secret sharing~\cite{CL21a, CL21b}, 
and this form is similar to multipartite private states, 
so it could be proven that these quantum states are genuinely multipartite entangled states.
However, it has been shown that genuine multipartite entanglement is not necessarily essential to obtain a nonzero key rate for conference key agreement~\cite{CKB21}.
The quantum state considered in Ref.~\cite{CKB21} is similar to the quantum state in Eq.~(\ref{thm3_state}).
It could be interesting to investigate whether the degree of separability can determine the limits of the key rate achievable for conference key agreement or quantum secret sharing.

\section*{ACKNOWLEDGMENTS}
M.C. acknowledges support from the National Research Foundation (NRF) of Korea grant funded by the Korea
Government (Grant No. NRF-2022M3K2A1083890).
E.B. acknowledges support from the NRF of Korea grant funded by the Ministry of Science and ICT (MSIT) (Grant No. NRF-2022R1C1C2006396).
G.L. acknowledges support from the NRF ok Korea grant funded by the MSIT (No. NRF-2022M3K2A1083859 and NRF-2022R1F1A1068197) 
and Creation of the Quantum Information Science R$\&$D Ecosystem (No. 2022M3H3A106307411) through the NRF of Korea funded by the MSIT.

\bibliography{GMEteleportation}

\end{document}